\documentclass[conference]{IEEEtran}
\IEEEoverridecommandlockouts
\usepackage[caption = false]{subfig}
\usepackage{amssymb}
\usepackage{amstext}
\usepackage{fixltx2e}
\usepackage{letltxmacro}
\usepackage{pifont}
\usepackage{multirow}
\usepackage{array}
\usepackage{color}
\usepackage{algorithm}
\usepackage{algorithmic}
\usepackage{nth}
\usepackage{amsmath}
\usepackage{array}
\usepackage{mathtools}
\usepackage{amsthm}

\ifCLASSINFOpdf
  \usepackage[pdftex]{graphicx}
\usepackage{epstopdf}
\else
  \usepackage[dvips]{graphicx}
\fi

\newtheorem{theorem}{Theorem}[section]
\newtheorem{lemma}[theorem]{Lemma}

\def\BibTeX{{\rm B\kern-.05em{\sc i\kern-.025em b}\kern-.08em
    T\kern-.1667em\lower.7ex\hbox{E}\kern-.125emX}}
\begin{document}

\title{How UAVs' Highly Dynamic 3D Movement Improves Network Security?\\
}

\author{\IEEEauthorblockN{1\textsuperscript{st} Mohammed Algharib}
\IEEEauthorblockA{\textit{School of Informatics, Computing, and Cyber Systems} \\
\textit{Northern Arizona University}\\
Flagstaff, AZ, US.\\
mohammed.algharib@nau.edu}
\and
\IEEEauthorblockN{2\textsuperscript{nd} Fatemeh Afghah}
\IEEEauthorblockA{\textit{School of Informatics, Computing, and Cyber Systems} \\
\textit{Northern Arizona University}\\
Flagstaff, AZ, US.\\
fatemeh.afghah@nau.edu}
}

\maketitle

\begin{abstract}
Cooperative ad hoc unmanned aerial vehicle (UAV) networks need essential security services to ensure their communication security. Cryptography, as the inseparable tool for providing security services, requires a robust key management system. Alas, the absence of infrastructure in cooperative networks leads to the infeasibility of providing conventional key management systems. Key pre-distribution schemes have shown promising performance in different cooperative networks due to their lightweight nature. However, intermediate decryption-encryption (DE) steps and the lack of key updates are the most concerning issues they suffer from. In this paper, we propose a simple and effective key management algorithm inspired by the idea of key pre-distribution, where it utilizes the highly dynamic UAV node movement in 3D space to provide the key update feature and optimizes the number of intermediate DE steps. Although it is a general model for any mobile ad hoc network, we have selected UAV network as an example domain to show the efficiency of the model given the high mobility. We define the communication density parameter to analytically show that using any highly dynamic random movement pattern leads our algorithm to work effectively. To show the proposed algorithm's effectiveness, we exhaustively analyze its security and performance in the UAV network using the ns-3 network simulator. Results validate our analytical findings and show how the highly dynamic UAV network movement helps our algorithm to provide the key update feature and to optimize the number of DE steps\footnote{This material is based upon work supported by the Air Force Office of Scientific Research under award number FA9550-20-1-0090 and the National Science Foundation under Grant Number CNS-2034218.}.
\end{abstract}


\section{Introduction}
Cooperative ad hoc UAV networks, also known as flying ad hoc networks (FANETs), have been increasingly deployed in different academic, commercial, and military applications. There are two general communication categories for UAV networks, infrastructure-based, and ad hoc \cite{comSurveyJNCA}. In the infrastructure-based UAV communication, it is suggested to use the already deployed infrastructure, such as a cellular network, to provide communication requirements. This technique has the potential to provide high-quality communication between the UAVs. However, in the case of unavailability of infrastructure, this technique obviously fails. The other technique is to provide communication with the cooperation of the nodes in an ad hoc mode. In this case, the quality of communication depends on the existence and the cooperation of other nodes. While having the UAVs as the aerial users with already deployed cellular networks is the most promising communication solution for UAV networks \cite{coexistance1,coexistance2}, cooperative ad hoc communication is required for the applications with no infrastructure access \cite{fanet,Khaledi}. In its unmanned aircraft systems roadmap for the years 2010-2035, the US army also emphasized the necessity of ad hoc communication for UAV networks \cite{roadmap}. Examples of such scenarios include wildfire monitoring, disaster relief, and search-and-rescue scenarios, to name only a few \cite{wildfire,relay}. 
In search and rescue missions, typically after a disaster or in areas with no infrastructure network coverage, the UAVs can be used as flying cameras to look/sense for a specific target \cite{searchRescue}. Sensing data such as pollution, humidity, temperature, light intensity, and pressure in hard to reach areas is another instance \cite{sensing1,sensing2}. Military applications such as collecting information on battlefield and law enforcement tasks are some other instances \cite{fanetApp}.

In most of the mentioned scenarios, the security of the transferred data is crucial. Like any other communication system, FANETs need to be provided with confidentiality, integrity, authentication, non-repudiation, and availability services \cite{roadmap}. Cryptography can play an essential role in providing security services. Regardless of being symmetric or asymmetric, any cryptosystem needs a robust key management system to deal with the key generation, key revocation, key update, key certificate, etc. The conventional key management systems are built based on an infrastructure or trusted third parties. There is no generally secure key management system that can work perfectly in an infrastructure-less environment, to the best of our knowledge. 

Some proposed key management systems are built based on the characteristics of the network of study in different ad hoc networks such as wireless sensor network (WSN), mobile ad hoc network (MANET), and vehicular ad hoc network (VANET). For instance, the idea of symmetric key pre-distribution is proposed for WSN to provide some chained encryption between communicating sensor nodes \cite{gligor}. Since the resource limitations are the main concerns in such networks, authors of \cite{gligor} proposed to store only $k$ randomly chosen keys in each node where $k<<n$ and $n$ is the number of network nodes. In this case, each node can communicate securely with its physically adjacent node if they share a common key. Elsewhere they have to find a key-path, which any adjacent nodes share a key. In this case, the transferred messages have to be decrypted at any intermediate node and encrypted again by the shared key with the next node. Such intermediate decryption and encryption processes are referred to as \textit{intermediate DE steps} and are obvious security threats. The resource limitations in WSNs lead to some security concerns such as intermediate DE steps, the absence of key update, and the infeasibility of providing essential security services such as integrity, digital signature, and non-repudiation to be completely ignored. 

Probabilistic asymmetric key pre-distribution (PAKP) \cite{waina} is another instance. It is proposed to be used in MANET where the nodes have less resource limitation compared to WSN. PAKP improves the security of conventional key pre-distribution schemes by decreasing the number of intermediate DE steps and taking advantage of asymmetric cryptography in providing the digital signature, integrity, and non-repudiation services. Instead, it increases the computation complexity of the required cryptographic operations. However, PAKP still cannot update the keys and carries the concern of intermediate DE steps.

In this paper, we propose an asymmetric key management system for cooperative UAV networks, inspired by the idea of key pre-distribution, which exploits the highly dynamic nature of FANETs to provide the key update feature and optimizes the number of DE steps. On the one hand, asymmetric cryptography requires higher complexity than symmetric cryptography, but it brings the feasibility of providing the digital signature, integrity, and non-repudiation security services. On the other hand, UAVs are practically limited by their energy in performing highly complex computations, which means that they can perform complex mathematical operations but have to perform such operations as less often as possible. Accordingly, in this paper, we use asymmetric cryptography and design our algorithm such that the UAVs need to perform asymmetric cryptography operations as less often as possible. 

We propose for each UAV to assign a specific storage for key management. Any pair of UAV nodes exchange their public keys when they become in the communication range. If the UAV has enough storage space, it stores the new key. Elsewhere, the UAV decides, based on a pre-defined strategy, to discard the key or replace it with a previously stored one. 

While the proposed idea could be used in any mobile cooperative network, the UAV high-speed 3D movement helps the proposed idea be more effective in FANETs. Hereby, we first define a parameter, refer to it as \textit{communication density}, to fairly compare FANETs with other cooperative networks such as MANET and VANET on its basis. Then, we analytically show that if the nodes move randomly, regardless of the movement pattern, the network density follows a Normal distribution. Consequently, the proposed algorithm will result in the same pattern for network density related parameters. We show via exhaustive simulation in network simulator ns-3 that the proposed algorithm is more secure and reveals higher performance in FANETs. Basically, the highly dynamic movement of FANET nodes in 3D space helps the algorithm work more efficiently compared to conventional MANET and VANET, in which the nodes move with less dynamicity in 2D space. 

To evaluate the proposed algorithm, we measure the number of intermediate DE steps, the connectivity probability, the storage required by each node, the lowest possible expiration time of the keys, the average time for each UAV to meet all other nodes, and the average time required to gain full connectivity, as the performance and security metrics. To report fair results, we perform our simulation using two random mobility models of random waypoint (RWP) \cite{rwp} and Gauss-Markov (GM) \cite{gm} mobility models. We further propose three different key storage and update strategies when the network nodes have limited storage space and cannot store all keys. 

The evaluation results show that, although the proposed algorithm works for conventional ad hoc networks, in ad hoc UAV networks it improves the network connectivity, reduces the intermediate DE steps by up to $25\%$, increase the probability of visiting all nodes by up to $75\%$, decrease the time to visit all nodes to half of the time, and reduces the time to provide full connectivity to the one third in comparison with the conventional ad hoc networks. The mentioned evaluation process covers a communication range starting from less than one node per communication area/sphere to 25 node. 

The main contribution of this paper is to propose a fully distributed key management algorithm, inspired by the idea of key pre-distribution and utilizes the highly dynamic nature of UAV network, to provide key update. Furthermore, we exhaustively analyze how network dynamicity can help to secure cooperative communication. We define a parameter referred to as communication density to perform a fair analysis on its basis. We analytically prove that using this parameter leads to a fairly identical pattern of results for different random mobility models. We further design, analyze, and compare three different strategies for the key updates. 

The rest of the paper is organized as the following. We review the related works in Section (\ref{sec::relatedWork}). Then, we propose the key management algorithm and the storage and update strategies in Section (\ref{sec::proposedAlg}). Section (\ref{sec::analysis}) reviews the simulation setting and includes the proposed algorithm's performance and security evaluation. Finally, the paper is concluded in Section (\ref{sec::conclusion}).  

\section{Related Work}
\label{sec::relatedWork}
Generally, there are two categories of key management systems that support the infrastructure-less mobile networks, distributed certificate-based and key pre-distribution based algorithms. Certificate-based algorithms aim at certifying nodes' public keys in a distributed manner \cite{certificate}. In such algorithms, a secret value is required for issuing the certificates. The secret value is distributed among the nodes using Shamir's secret sharing algorithm \cite{shamir}. In this algorithm, the secret is shared among $n$ nodes such that any $k$ of them can certify the public keys. Indeed, any $k-1$ shares cannot reveal any information about the secret value. While the general idea is interesting, it is vulnerable against the node compromise attacks. The attacker needs to capture only $k$ nodes to compromise the entire network. Increasing the value of $k$ leads to a higher level of security. However, it might make the certifying process infeasible.        

Key pre-distribution idea is proposed in two completely different versions for symmetric and asymmetric cryptosystems. It is originally proposed by Eschenauer and Gligor \cite{gligor} for the symmetric cryptosystem for WSNs. The basic idea is to store only $k$ keys in each node before the network starts to work, where $k<<n$ and $n$ is the number of network nodes. The keys are chosen uniformly in random with replacement from a key-pool containing all keys. Any physically adjacent nodes that share a common key can use the key and communicate securely. Elsewhere, they have to find a key-path. 

While the basic idea is simple, it adds some level of security to the WSN with a large number of nodes and limited resources. However, it suffers from several security concerns. First, it is vulnerable against the node capture. The attacker can get access to $k$ secret keys by capturing each node. Hence, the attacker can compromise the entire network by capturing several nodes. One idea to face this shortcoming is to use larger key-pools, i.e. key-pool with a larger number of keys. However, the larger key-pool decreases the probability of sharing keys and consequently leads to network disconnection. The other idea is to make a secure connection between adjacent nodes if and only if they share at least $q$ keys \cite{qcomp}. However, this idea leads also to network disconnection for large $q$ values. There are many other works proposed to face this shortcoming \cite{polynomial,sst,bibd,linkSecurity3,linkSecurity,linkSecurity4,linkSecurity2}.

The second security shortcoming is the large number of intermediate DE steps. Using the Eschenauer-Gligor idea, the path from the source node to the destination has to be a physical path, in which any adjacent nodes share a common key. Hence, the lower bound of the number of intermediate DE steps in this algorithm is equal to the physical path length. Furthermore, since the probability of sharing a secret key is less than one, the number of DE steps is much more in practice. This fact also leads to significant performance degradation due to the generally longer path length. To face this obvious security and performance shortcoming, many works have proposed to guarantee the shared key between the adjacent nodes \cite{bibd,deterministic1,deterministic2}. Even using such algorithms, the number of intermediate DE steps stays in the order of physical path length, which is high. One practical solution is to use some disjoint key-paths to perform a key-exchange process. The data transmission process then uses the exchanged key with the destination to make an end-to-end secure communication \cite{multipath1,multipath2,multipath3}. The mentioned multi-path solutions improve the security of the symmetric based key pre-distributions. However, they are still vulnerable to cooperative attacks or multiple node capture.   

PAKP \cite{waina} is the first asymmetric key pre-distribution introduced to be used in MANETs. The core idea is similar to the Eschenauer-Gligor idea. However, the key-pool includes the public key of all nodes instead of random secret keys. Each node is pre-loaded by $k$ public keys, chosen uniformly at random from the key-pool, before the network starts to work. Clearly, the first superiority of PAKP in comparison with the symmetric key pre-distributions is in the distribution of public values. By capturing any node, the attacker can get access to $k$ public keys and only one private key of the compromised node. Second, in PAKP, the source node finds a key-path to the destination, and then for each hop in the key-path, it has to find the corresponding physical path. It means that the intermediate DE steps happen only in the intermediate nodes of the key-path. Authors of \cite{waina} showed that the number of intermediate DE steps, in this case, is in the order of $(\log_k n)$ which is much lower than that of symmetric key pre-distribution.     

Although the authors of \cite{waina} showed that the number of intermediate DE steps in their algorithm is too low, any DE step stays an obvious security threat. In \cite{kpsec}, authors proposed an algorithm that utilizes the disjoint key-path of PAKP to exchange the public keys of the source and the destination before data transfer. The data is then encrypted and transferred using the public key of the destination. Hence there will not be any intermediate DE steps.     

As we mentioned earlier, in symmetric key pre-distribution, a physical path has to be found first. The corresponding key-path follows each hop of the physical path. In asymmetric key pre-distribution, the routing process is in reverse; thus, the key-path has to be found first, then, for each hop in the key-path, the corresponding physical path is used. For each step of routing, any conventional routing algorithm could be used. Hence, regardless of being symmetric or asymmetric, a conventional routing algorithm has to be used twice to find the general path for the key pre-distribution. Authors of \cite{tmc} showed that such a routing process of using a conventional routing algorithm twice leads to a non-optimal path. They proposed an optimization problem that finds the generally optimal path.   

Obviously, using key pre-distribution in any of its versions, symmetric or asymmetric, there is no option for revoking and updating the keys. Hence, the compromised key stays compromised until the end of the network lifetime. In this paper, we aim at utilizing the highly dynamic nature of cooperative UAV networks to add the feature of expiring and updating the keys.  

\section{Proposed Algorithm}
\label{sec::proposedAlg}

The algorithm of this paper is inspired by the general idea of key pre-distribution schemes. However, we propose to collect the keys during the network operation times, instead of pre-distributing them before the network starts to work. In this section, we first review the key generation and exchange process. Then, we define a communication density parameter which gives us the ability to fairly analyze the network parameters on its basis. Finally, we review some key store and update strategies.   

\subsection{Key Generation and Update}
To make the UAVs able to update their keys, we propose for each node to pick an arbitrary private key and then generates its corresponding public key. While any well-known asymmetric cryptosystem can be utilized, we suggest to use elliptic curve cryptography (ECC) for its shorter key-length and less computational complexity for the same-level of security. 

The ECC system parameters $(p,a,b,G,n,h)$ has to be set before the network start to work. These parameters are required to form the elliptic curve $y^2=x^3+ax+b$ as a plane curve over the finite field $\mathbb{F}_p$, and to help the nodes to generate their public-private key pairs. $G$ is the generator of the cyclic subgroup and known as base point.While $n$ is the smallest positive number satisfying $nG=\mathcal{O}$ and $\mathcal{O}$ is a point at infinity, $h$ is a small integer number satisfying $h=\frac{1}{n}|E(\mathbb{F}_p)|$.  

Any node can choose any arbitrary value $x \in [1\quad n-1]$ as its private key. The corresponding public key is then $y=x.G$, which means adding $G$ to itself for $x$ times. Any UAV node, by generating its private-public key pairs, chooses an expiration time for this pair of keys. Any pair of UAV nodes that get into the communication range of each other can exchange their public keys safely along with their signature. The signature is a certificate signed by the private key of the node to certify the ownership of the private key for the exchanged public key. Along with the public key, each node sends the key's expiration time to the other node. It is out of the context of this paper to go through the hardware and software required to improve the security of the generated keys. However, we suggest using physically unclonable functions (PUF) with technologies like Resistive Random Access Memory (ReRAM) and static random-access memory (SRAM) \cite{puf,Korenda} in the key generation phase, to make it feasible to provide security services such as non-repudiation. 

Any pair of nodes that have their public keys, can communicate securely before the key gets expired. Any pair of nodes that have not met earlier or their stored keys have expired, can find a \textit{key-path} to communicate with one another. A key-path is a path starting from the source node and ending by the destination in which each adjacent nodes have their public keys. While the idea of key-exchange during network operation time is not a novel idea, the conventional cooperative networks need a long time for key collection and exchange to make all nodes able to communicate with each other. Accordingly, they cannot update the keys regularly and the key update leads to network disconnection. We show this fact in Section (\ref{sec::analysis}). 

The other problem is that the intermediate nodes in the key-path have to decrypt and then encrypt the forwarded packets. This obvious security imperfection referred to as intermediate DE steps. Although the problem of intermediate DE steps will not be totally resolved in this paper, we show that the highly dynamic nature of FANETs leads our key management algorithm to efficiently decrease the number of intermediate DE steps as well as the time required for the nodes to get the public keys of one another. We show further that the connectivity of this algorithm converges very fast which helps the nodes to decrease the expiration time of their keys and consequently reach higher security milestones.

\subsection{Communication density}

In one hand, the proposed key exchange process relies on the node movement pattern in the network. On the other hand, we aim at calculating the key-path existence probability, time to full connectivity, number of intermediate DE steps, probability of visiting all other nodes, and time to visit all, to show how the highly dynamic nature of UAV networks helps our algorithm to work more effective in FANETs in comparison with the conventional cooperative networks. To make a fair comparison we need a single parameter to examine the network performance on its basis. Hence, we define  \textit{communication density} parameter. It is the average number of nodes in the communication area and communication sphere covered by the communication range of a node in 2D area and 3D sphere, respectively. Since in FANETs, the UAVs move in 3D sphere, we calculate the communication density by dividing the number of UAV nodes over the communication sphere density of one UAV. Assume that the communication range of the UAV is $r$, there is a sphere with the volume of $\frac{3}{4}\pi r^3$ which represents the communication sphere of that UAV. Any other UAV gets into this sphere will be in the communication range of the corresponding UAV. Hence, the communication density represents the average number of nodes in the communication area/sphere of the node. For MANET/VANET, the communication area of the node is $\pi r^2$. Hence, the communication density is $n/(\frac{XY}{\pi r^2})$, where $n$ is the number of network nodes. Fig. (\ref{fig::commDensity}) conceptually shows  the communication density parameter for 2D area and 3D sphere. 

\begin{figure}[t!]
	\centering
	\subfloat[3D space]{\includegraphics[trim=1.5cm 0 3cm 0,width=0.5\linewidth]{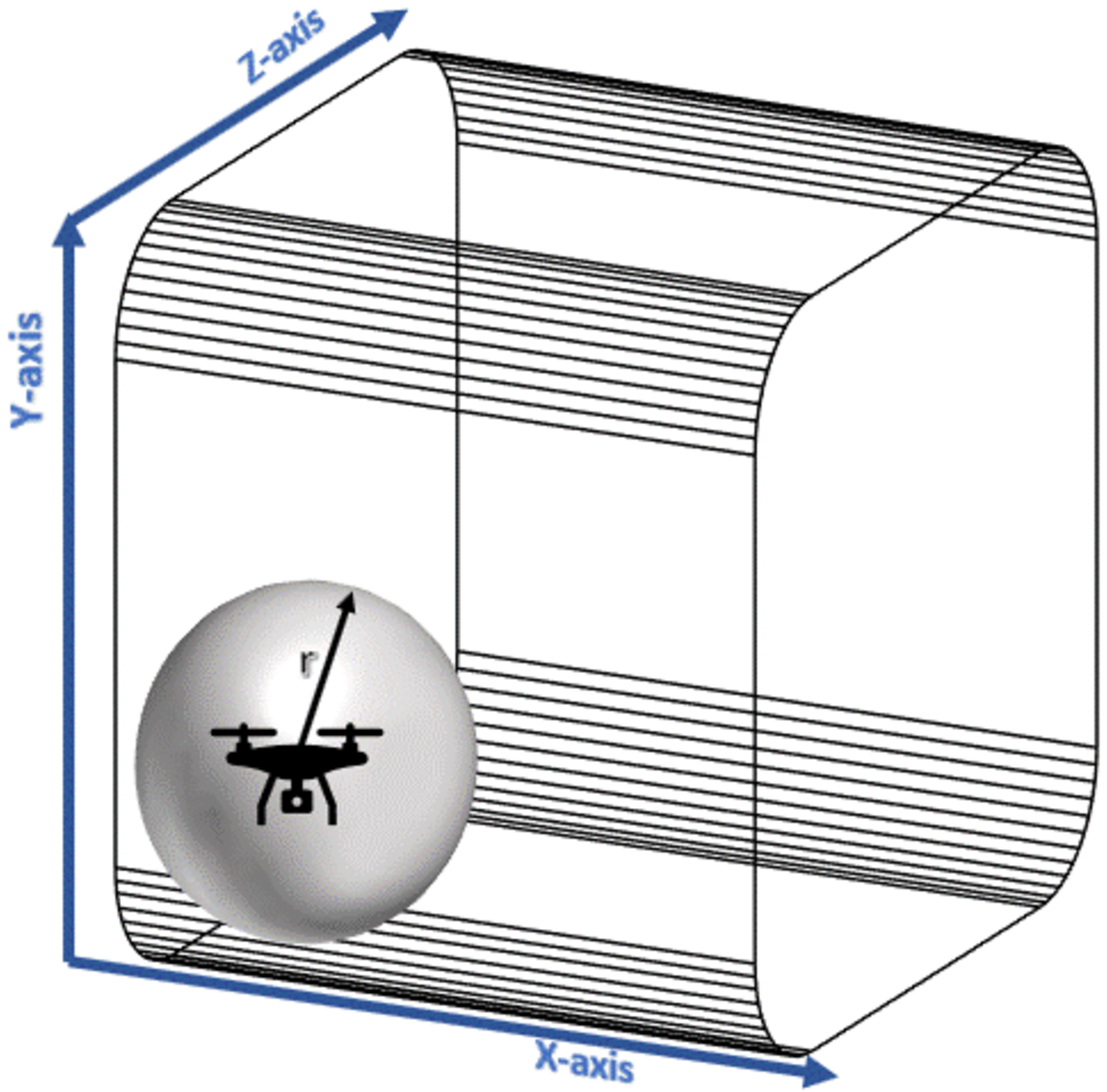}}
	\subfloat[2D area]{  \includegraphics[trim=2cm 0 1cm 0,width=0.5\linewidth]{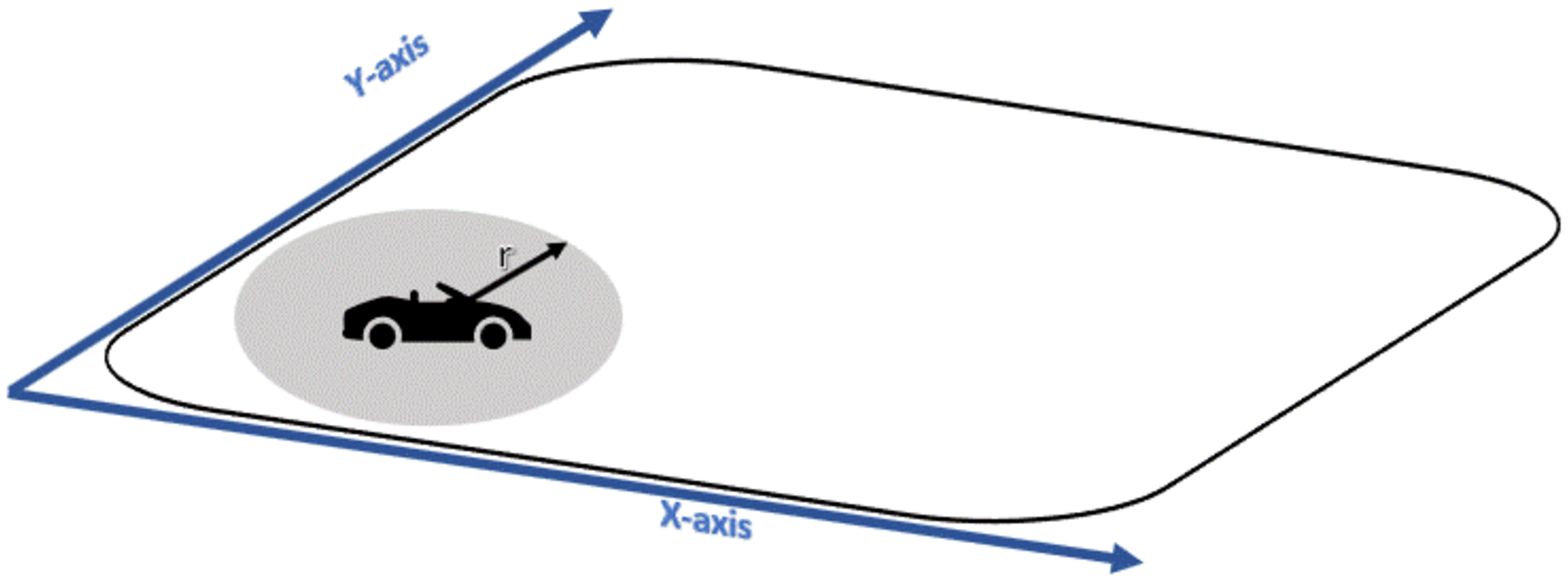}}
	\caption{Communication density}
	\label{fig::commDensity}
\end{figure}

As it is clear, the mentioned security and performance parameters are sensitive to the pattern of nodes' movement. Accordingly, we aim at showing that the pattern of the variation of these parameters based on the network density is fairly identical for any random mobility model. 

\begin{lemma}
\label{lem::nodeDistribution}
If the nodes move independently according to an identical random movement algorithm, the distribution of the nodes' density in each time snapshot follows the Normal distribution with the expected value $\mu$ and variance $\sigma^2$. The values of $\mu$ and $\sigma^2$ may vary based on the mobility models.
\end{lemma}

\begin{proof}
We use central limit theorem (CLT) to prove this lemma. Let us represent the movement of each node by random variable $\textbf{X}$. $\textbf{X}$ is a random variable in three-dimentional space $\textbf{X}=(X,Y,Z)$ where $X$, $Y$, and $Z$ are random variable themselves. According to the lemma's assumption, the nodes move independently according to an identical random movement algorithm. Thus, we represent the movement of each node by $\textbf{X}_i$ where $i=0,1,\dots,n-1$, and $n$ represents the number of network nodes. We know that each $\textbf{X}_i$ follows the same distribution, let us say distribution $\delta$, i.e. $\textbf{X}_i\sim\delta(\mu_\delta,\sigma^2_\delta)$. In one hand, the sample mean of the i.i.d. $\textbf{X}_i$s represents the distribution of the network nodes' density. On the other hand, according to CLT, the sample mean of i.i.d. random variables follows Normal distribution with the expected value $\mu=\mu_\delta$ and variance of $\sigma^2=\frac{1}{n}\sigma^2_\delta$. Hence, if the nodes move independently according to an identical random movement algorithm, the distribution of the nodes' density in each time snapshot follows Normal distribution with the expected value $\mu=\mu_\delta$ and variance $\sigma^2=\sigma^2_\delta$.
\end{proof}

Lemma (\ref{lem::nodeDistribution}) reveals that for the same network setting but different mobility models, the values of network parameters such as the number of neighbors, connectivity probability, time required for full connectivity, and time required to visit all nodes may vary for different mobility models based on the distribution of the movement of each specific model. However, the pattern of the mentioned parameters stays fairly identical for different communication densities. We investigate this fact via exhaustive simulation in Section (\ref{sec::analysis}). It is worthy to mention that even for small $n$ values, the CLT is proved to work with a high confidence level in practice.

\subsection{Key Storage Strategies}
Clearly, the best strategy for key storage and update is to store the public key of any met node. Accordingly only the expired keys are subject to remove. However, for the case that the nodes have limited storage space or the network is very large scale, the UAV cannot store the public key of any node that gets into its communication range. Thus, an efficient algorithm for storage management is required. In the following, we discuss three key storage strategies.

Assume that we have storage space for $k$ keys where $k<n$, and $n$ is the number of all nodes. In the first strategy, we start storing the public keys until the storage becomes full. Then, we replace any new key with the stored key that have the lowest key lifetime, i.e. the key with closest expiration time. In this case, we always have the most possible fresh key-table. In the second strategy, the UAV starts storing the public key of other UAVs that are located in its communication range until its storage space becomes full. In this strategy, we replace only the expired keys with new ones. Hence, the UAV's key-table is not updated unless a key gets expired. We define the third strategy to be a combination of the first and the second strategies, to possibly carry their advantages. In this strategy, we divide the storage into two separate storage with $k_1$ and $k_2$ key storage space, where $k_1+k_2=k$. Each UAV starts storing the public keys of the visited nodes until its storage becomes full. Then, it keeps updating the first $k_1$ keys where the rest of keys stay as they are until get expired. 

Although the proposed key storage strategies are simple, we show that the different strategies have significant effect on the performance and security parameters of the network. Section (\ref{sec::analysis}) represents the results of our analysis. 

\section{Security and Performance Analysis}
\label{sec::analysis}

In this section, we first review the simulation setting. We then compare the performance and security metrics of our proposed algorithm between a UAV network and a conventional mobile cooperative network such as MANET or VANET. Finally, we compare the effect of the proposed key storage mechanisms on the performance and security metrics. 

\subsection{Simulation Setting} 
We use network simulator ns-3 \cite{ns3} to perform our evaluation. While Table (\ref{tbl::simulationSetting}) represents the simulation setting, we consider the probability of key-path existence, key-path length, probability of visiting all other nodes, time to visit all, and time to full connectivity as the comparison metrics.

\begin{table}[t]
\caption[]{Simulation Setting }
\resizebox{1\textwidth}{!}{
\begin{minipage}{\textwidth}
\begin{tabular}{ l | l  }
 Simulator version & ns-3 3.30\\
  Number of nodes & $100$ \\
  Area size& $X\times Y\times Z m$\\
  Range of $X$ and $Y$ &$[500:100:1500] m$\\
  Maximum elevation for FANET& $100 m$\\
  Transmission power & 7.5 dBm\\
  Simulation time & 1000 sec\\
  Speed range for FANET& $[0 \quad 50] m/s$\\ 
  Speed range for VANET/MANET& $[0 \quad 20] m/s$\\ 
  Mobility models & RWP and G-M\\
  Azimuthal range in G-M model & $[0 \quad 0.05]$ rad\\
  Wireless communication standard & IEEE 802.11b\\
  Propagation loss for FANET  & Free-space propagation loss\\
  Propagation loss for VANET/MANET  & Two-ray propagation loss\\  
  Propagation delay model & Constant speed propagation delay
\end{tabular}
\label{tbl::simulationSetting}
\end{minipage}}
\end{table}

Generally, there are three main differences between FANET and MANET/VANET. First, in FANETs the UAV node moves in 3D sphere in comparison with the movement in MANET/VANET where the node moves in 2D area. Second, the communication in FANET is often considered as the line of sight and hence the propagation loss model is a free space loss model, known also as Friis model \cite{speed}. In MANET/VANET the propagation loss is assumed to follow a two-ray propagation loss model. Third, the dynamicity of the nodes in FANET is known to be much more than the MANET/VANET. To apply a higher dynamicity to UAVs, we consider the maximum speed for MANET/VANET to be  $20 \frac{m}{s}$, while we assume that UAV's maximum speed is equal to $50 \frac{m}{s}$ \cite{speed}. 

To keep the simulation fair, we set the maximum elevation of UAV nodes to $100 m$, and consider a square area size with the length ranging from $500m$ to $1500 m$ in the increment steps of $100m$. Using the mentioned parameters, the communication density for FANET and VANET/MANET stays almost equal for different area lengths. We consider two well-known mobility models, random waypoint (RWP) \cite{rwp} and Gauss-Markov (GM) \cite{gm} to cover different mobility patterns. While RWP is a fairly random mobility model, GM mobility model prevents the nodes from sudden significant variation in the angle of movement. Since in real world the nodes try to reach a destination and hence do not make a significant change in their directions, this limitation leads to more realistic scenarios. We performed the simulations 20 times with different random seed values. Each time every parameter is calculated in  1 sec snapshots. The reported results are the average of all the traced values.

\subsection{FANET versus MANET/VANET Comparison}

In this section, we aim at comparing the performance and the security of the proposed algorithm in FANET compared to MANET/VANET. Hereby, we consider a large enough storage space for the nodes to keep any keys. This assumption helps us  evaluate some parameters such as the time to reach full connectivity or the time to visit all other nodes. Fig. (\ref{fig::KPprobUnlimited}) compares the probability of key-path existence for different communication densities. This figure considers all the combinations of FANET and MANET/VANET networks, as well as RWP and GM mobility models. While as the density increases, the key-path probability converges to one for all combinations, its convergence speed is much faster in UAV networks for both the considered mobility models. In this figure, the upper horizontal axis represents the communication density in 2D area where the lower horizontal axis represents the communication density for 3D sphere.  

\begin{figure}[t!]
	\centering
	\includegraphics[width=\linewidth]{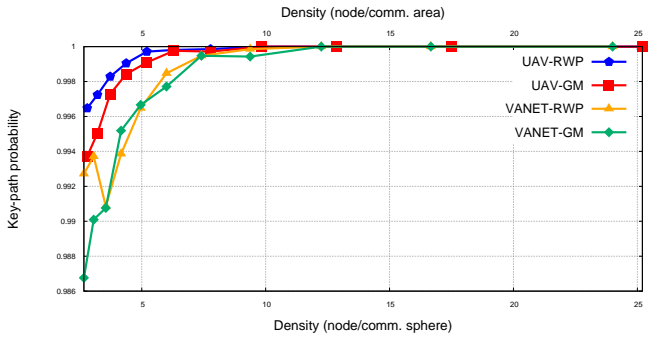}
		\caption{A comparison of the probability of key-path existence between FANET and MANET/VANET.}
	\label{fig::KPprobUnlimited}
\end{figure}

Next, we calculate the time to reach the full key-connectivity. This parameter shows us the minimum feasible time for key expiration. If we set the key expiration time lower than this value, the nodes will not have enough time to find a secure connection to one another. Fig. (\ref{fig::TTFC}) shows the time required for the network in which all nodes become able to find a key-path to one another. As it is clear in this figure, FANET significantly outperforms other networks. In the considered simulation setting, the time required for the full key-path connectivity for FANET with GM mobility is about 12 second which is one third of the time required by other conventional cooperative networks.   

\begin{figure}[t!]
	\centering
     \includegraphics[width=\linewidth]{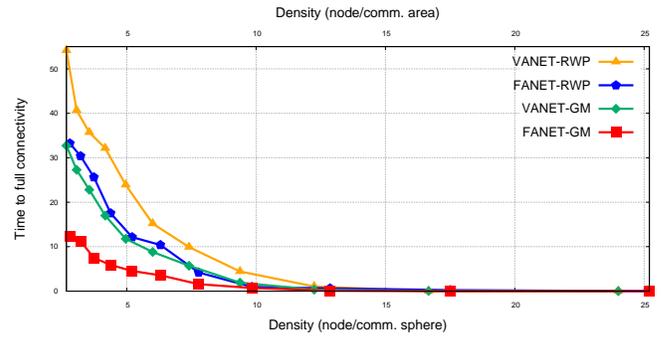}
	\caption{A comparison of time required for the network to be fully connected via key-paths for FANETs versus VANETs, both RWP and GM mobility models.}
	\label{fig::TTFC}
\end{figure}

The results of Fig. (\ref{fig::KPprobUnlimited}) ensure us that the key exchange algorithm provides the key-connectivity for the network. However, the key-path length represents the number of intermediate DE steps and considered as a security as well as performance metric. Fig. (\ref{fig::DEUnlimited}) represents the number of intermediate DE steps for all combinations of FANET versus MANET/VANET as well as RWP versus GM mobility models. This figure shows a significant improvement in the number of intermediate DE steps in FANET in comparison with conventional cooperative networks. For the FANET with GM mobility, this parameter is always close to one, which means that in most cases only one intermediate node is included in the key-path.   

\begin{figure}[t!]
	\centering
  \includegraphics[width=\linewidth]{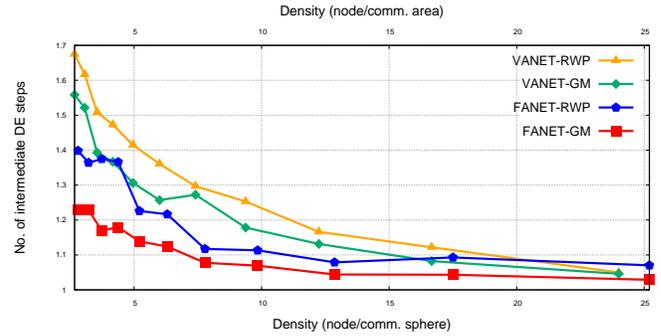}
	\caption{A comparison of the number of intermediate DE steps.}
	\label{fig::DEUnlimited}
\end{figure}

Fig. (\ref{fig::visitAllprob}) shows the probability of visiting all other nodes. This parameter is the main parameter that shows the effectiveness of our algorithm on FANET in comparison with other conventional cooperative networks. Here, we can see how the high dynamicity of UAV movement in 3D space leads our proposed algorithm to be effective. In addition to the probability of visiting all other nodes, the time of visiting all is of crucial. Fig. (\ref{fig::visitAllTime}) compares this parameter for different networks and different mobility models. This parameter helps us figure out how to manage the expiration time of the stored keys to keep the network connected. As it is clear from this section's figures, the shape and the patterns of all curves in each figure are almost identical, which  validates our analytical proof of Lemma (\ref{lem::nodeDistribution}).

\begin{figure}[t!]
	\centering
     \includegraphics[width=\linewidth]{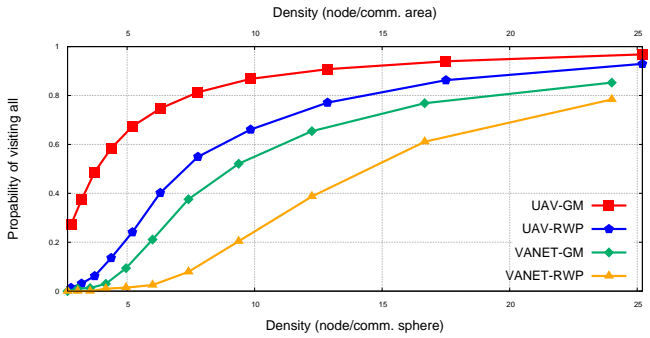}
		\caption{Probability of visiting all other nodes.}
	\label{fig::visitAllprob}
\end{figure}
	
	\begin{figure}[t!]
	\centering
    \includegraphics[width=\linewidth]{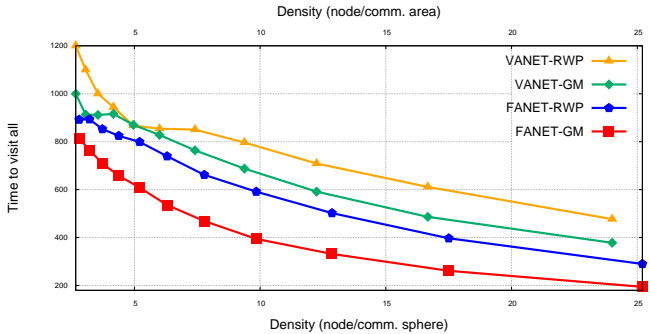}
	\caption{Time to visit all other nodes.}
	\label{fig::visitAllTime}
\end{figure}

\subsection{Key Storage Strategies Comparison}

As we proposed in Section (\ref{sec::proposedAlg}), we consider three strategies for key storage and update. In the first strategy, the UAV nodes simply exchange and update their keys whenever they get into the communication range of one another. Thus, any new key takes place the expired key or the key with the closest expiration time. In the second strategy, we store the key of the newly met UAV if and only if we have a unused storage space or an expired key. The third strategy is a combination of both the mentioned strategies. Hence, the storage space of each UAV node is divided into two parts. The first part stores any new keys and replaces the oldest one if this part of storage is full. The second part never replaces a key except when there is an expired one. We considered a storage space for 10 keys in each UAV. For the last strategy, we considered each part with 5 key storage spaces, i.e. $k_1=k_2=5$. We consider $100$ seconds as the expiration time for each key. To compare the proposed strategies for their performance and security, we report the key-path existence probability, average number of intermediate DE steps, and overall path length. 

First, we consider the key-path existence probability. This parameter represents how much the proposed algorithm is effective and practical. Fig. (\ref{fig::keyPath}) shows the results for the combination of three proposed strategies and both the RWP and GM mobility models. The worthy mentioning outcome of this figure is that for both mobility models at the density of 7, the network is almost connected by the key-paths. However, the first  strategy represents the best results for this parameter.

\begin{figure}[t!]
	\centering
	\subfloat[RWP mobility model.]{\includegraphics[width=\linewidth]{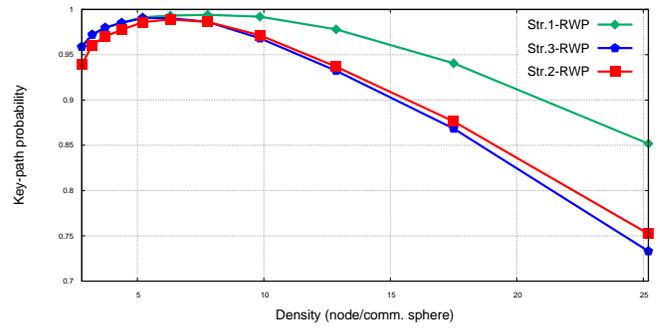}}\\
	\subfloat[GM mobility model.]{  \includegraphics[width=\linewidth]{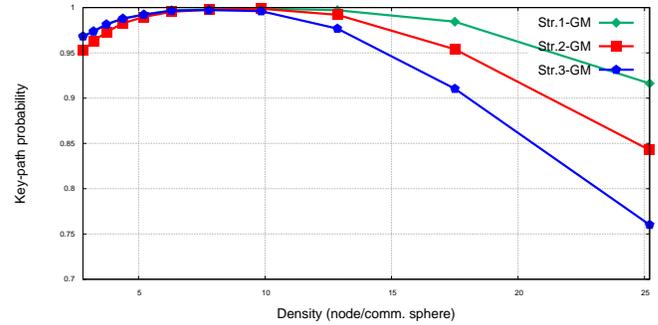}}
	\caption{A comparison of key-path existence probability for FANETs with limited storage.}
	\label{fig::keyPath}
\end{figure}

Fig. (\ref{fig::keyPathLen}) shows the results for the average number of intermediate DE steps. Obviously the less number of intermediate decryption and encryption means the less chance for the attacker to get access to the transferred data by compromising the UAV nodes. As Fig. (\ref{fig::keyPathLen}) shows, the faster key-update in cooperative UAV network leads to less number of intermediate DE steps. Accordingly, the first strategy has the lower number of DE steps, and the hybrid strategy outperforms the second one.

\begin{figure}[t!]
	\centering
	\subfloat[RWP mobility model.]{\includegraphics[width=\linewidth]{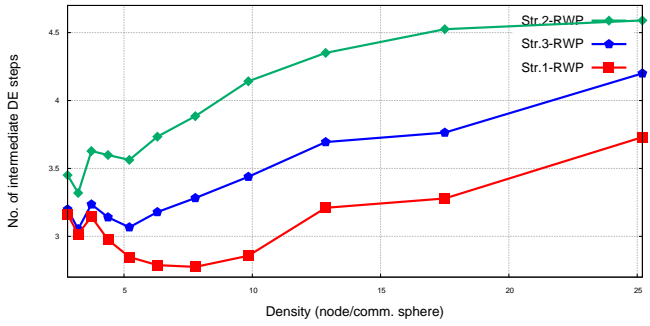}}\\
	\subfloat[GM mobility model.]{  \includegraphics[width=\linewidth]{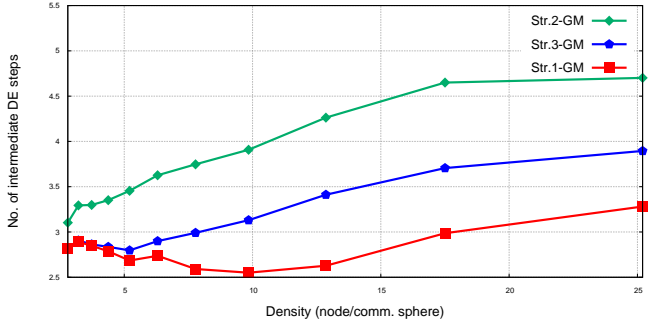}}
	\caption{A comparison of the number of intermediate DE steps for FANETs with limited storage.}
	\label{fig::keyPathLen}
\end{figure}

As we mentioned earlier, using the proposed idea, the routing process has to find the key-path first and then for each hop in the key-path the corresponding physical path is required. Hence, the overall path length is always greater than or equal the key-path length. It is worth mentioning that there are different key-exchange protocols which utilize the key-paths to make a key-agreement between the communication end parties. Hence, the final data communication exploits the shortest physical  path. However, the overall path is used for key exchange process. Fig. (\ref{fig::overalPath}) shows the overall path length for the proposed strategies and both mobility models. The randomness of RWP seems to have a positive impact on the performance of the algorithm in comparison with the more realistic GM mobility model. Again, the first key update strategy shows the highest performance among all considered strategies. It is notable that for the network with density more than 10, the overall path length stays almost steady. Looking at the results of this figure and Fig. (\ref{fig::keyPathLen}) at the same time, we can conclude that the increment in the number of DE steps, i.e. increment in the key-path length, in the more dense networks does not increase the overall path length.

\begin{figure}[t!]
	\centering
	\subfloat[RWP mobility model.]{\includegraphics[width=\linewidth]{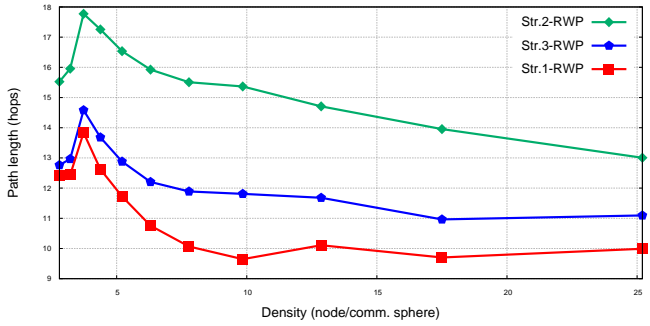}}\\
	\subfloat[GM mobility model.]{  \includegraphics[width=\linewidth]{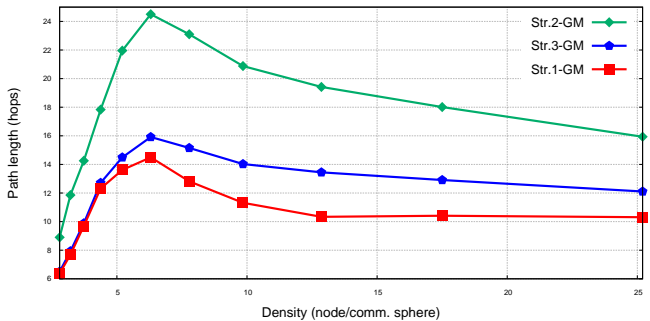}}
	\caption{A comparison of the overall path length.}
	\label{fig::overalPath}
\end{figure}

\section{Conclusion}
\label{sec::conclusion}
Inspired by the idea of key pre-distribution and by utilizing the highly dynamic nature of UAV networks, in this work, we proposed a practical key management with active update feature and optimized  number of DE steps. We analytically proved that for any random movement, the pattern of network density follows Normal distribution where the expected value and the variance vary based on the pattern of movement and area size. Then we validate the analytical proof by exhaustive simulations in ns-3 network simulator. We showed how the highly dynamic nature of UAV networks helps the proposed key management algorithm to work effectively. Finally, we compared different strategies for key storage and update and see that the regular update in the key-table leads to the highest performance and security for the measured parameters.

\nocite{*}
\bibliographystyle{IEEEtran}
\bibliography{References}

\end{document}